\documentclass[10pt,conference]{IEEEtran}
	\IEEEoverridecommandlockouts
	\makeatletter

\usepackage[ruled,linesnumbered,vlined]{algorithm2e}
\usepackage{amsthm}
\usepackage{epsfig}
\usepackage{amssymb}
\usepackage{amsmath}
\usepackage{amsfonts}
\usepackage{verbatim}
\usepackage{setspace}
\usepackage{subfigure}
\usepackage{url}
\usepackage[left=.55in,right=.55in,top=.63in,bottom=.5in]{geometry}

\newtheorem{theorem}{Theorem}
\newtheorem{lemma}{\hspace{0pt} Lemma}

\usepackage{fancyhdr}
\fancypagestyle{empty}{
  \fancyhf{}
  \fancyhead[C]{\footnotesize{Q. Liu, T. Luo, R. Tang, and S. Bressan, ``An efficient and truthful pricing mechanism for team formation in crowdsourcing markets'', {\it ICC 2015}, pp. 567-572.}}
}

\title{An Efficient and Truthful Pricing Mechanism for Team Formation in Crowdsourcing Markets}

\author{
Qing Liu$^*$, Tie Luo$^{\dagger}$, Ruiming Tang$^*$,  St{\'e}phane Bressan$^\S$\\

\fontsize{10}{10}\selectfont\rmfamily
$^*$School of Computing, National University of Singapore\\
$~^{\dagger}$Institute for Infocomm Research, A*STAR, Singapore\\
$^\S$National University of Singapore and CNRS, Image \& Pervasive Access Lab (IPAL UMI), Singapore\\

\fontsize{9}{9}\selectfont\ttfamily\upshape
\{liuqing,tangruiming,steph\}@nus.edu.sg, luot@i2r.a-star.edu.sg
}

\begin{document}
\maketitle
\thispagestyle{empty}

\begin{abstract}
In a crowdsourcing market, a requester is looking to form a team of workers to perform a complex task that requires a variety of skills. Candidate workers advertise their certified skills and bid prices for their participation.  We design four incentive mechanisms for selecting workers to form a valid team (that can complete the task) and determining each individual worker's payment. We examine profitability, individual rationality, computational efficiency, and truthfulness for each of the four mechanisms. Our analysis shows that TruTeam, one of the four mechanisms, is superior to the others, particularly due to its computational efficiency and truthfulness. Our extensive simulations confirm the analysis and demonstrate that TruTeam is an efficient and truthful pricing mechanism for team formation in crowdsourcing markets.
\end{abstract}


%
\IEEEpeerreviewmaketitle

\section{Introduction}
Future crowdsourcing platforms need to support collaboration ~\cite{kittur2013future}. In a collaborative crowdsourcing market, a requester is looking to form a team of workers who can perform a task that requires a set of skills. Interested candidate workers advertise their certified skills and bid prices for their participation. This differs from usual team formation problems (e.g.~\cite{anagnostopoulos2010power,anagnostopoulos2012online,lappas2009finding}) in that it not only considers the skills required by the task and possessed by the workers, but also involves economic incentives and several criteria that guarantee profitability for requesters and workers, social welfare, and truthfulness~\cite{nisan2007algorithmic}. A large body of research on crowdsourcing markets have confirmed the intuition that financial incentives increase workers' interest~\cite{DiPalantinoV09} and effort level~\cite{chen2010knowledge} as well as the attractiveness to experienced workers~\cite{mason2010financial}. However, to the best of our knowledge, existing pricing models (e.g.~\cite{DBLP:conf/www/SingerM13,DBLP:conf/www/SinglaK13,goel2014allocating}) for crowdsourcing platforms only consider individual workers without taking into account teamwork which is crucial in collaborative environments~\cite{kittur2013future}.

In this paper, we design task allocation and pricing mechanisms for selecting a team of workers in crowdsourcing markets. We start by presenting two baseline mechanisms which pay the selected workers the same as their bids.  The first baseline mechanism looks for the lowest-cost team by enumerating all possible teams in a brute force manner. This mechanism is \textit{profitable} for both the requester and selected workers but is not computationally efficient. It also does not ensure {\em truthfulness} or \textit{incentive-compatibility}, which means that workers may not bid their true costs but can cheat to gain higher payments. We refer to this baseline mechanism as \textsf{OPT} as it minimizes total payment if all workers bid truthfully. We show that \textsf{OPT} is profitable, individually rational but not efficient nor truthful.

The second baseline mechanism follows a greedy approach to form a team of workers with low bids and high total expertise. We refer to this mechanism as \textsf{GREEDY} and show that it is efficient, profitable, individually rational but not truthful.

Both baseline mechanisms do not prevent workers from overbidding, which necessitates the need of designing mechanisms in which workers will bid their true costs.

We first adapt the celebrated Vickrey-Clarke-Groves auction to our problem, and refer to this mechanism as \textsf{VCG}. We show that it is profitable, individually rational, truthful but not efficient.

Finally, we design a mechanism that combines the greedy selection rule and a special payment scheme. A selected worker receives a payment that equals to the highest bid she could have placed and been selected. We refer to this mechanism as \textsf{TruTeam}. It possesses all the above desirable properties, i.e., efficiency, profitability, individual rationality, and truthfulness.

Using synthetic scenarios, we evaluate the properties and performance of these four mechanisms. The results show that \textsf{TruTeam} is an efficient, truthful task allocation and pricing mechanism for team formation in crowdsourcing markets. In summary, this paper makes the following contributions:
\begin{itemize}
\item To the best of our knowledge, this is the first study on team formation in collaborative crowdsourcing markets.
\item We formulate the problem of team formation in crowdsourcing as an task allocation and pricing mechanism design problem.
\item We design two baseline and two refined mechanisms and prove profitability, individual rationality, prove or disprove truthfulness for each of the four mechanisms. In addition, we also evaluate their computational complexity.
\item We show via both analysis and extensive simulations that \textsf{TruTeam} is an efficient, profitable, individually rational and truthful mechanism for team formation in crowdsourcing markets.
\end{itemize}

\section{Related Work}
%
%

\subsection{Pricing Mechanisms}

Various models and techniques have been proposed for pricing on crowdsourcing platforms. Budget-feasible mechanisms \cite{singer2010budget} maximize a requester's profit under a budget constraint while satisfying other properties such as truthful bidding. In \cite{DBLP:conf/www/SingerM13}, the authors designed a framework for task allocation and pricing in an online environment. The framework aims to maximize the number of tasks performed under a given budget or to minimize payments for a given number of tasks.
The authors of~\cite{DBLP:conf/www/SinglaK13} designed a no-regret posted price mechanism which bridges the gap between procurement auctions and multi-armed bandits. That mechanism satisfies budget feasibility, achieves near-optimal utility for the requester, and also guarantees that workers bid their true costs.
In~\cite{goel2014allocating}, workers and tasks are modeled as a bipartite graph where an edge $(m,n)$ in the graph represents worker $n$ is willing to perform task $m$. The authors designed a payment mechanism that ensures budget feasibility and one-way-truthfulness while achieving near-optimal utility. Recent work \cite{profit14infocom} considers variable rather than fixed payment in crowdsourcing, as a function of the best worker's effort, and achieves a higher utility than the optimum achieved by fixed payment. This idea was then extended to heterogeneous crowdsourcing environments \cite{hetero14mass} to tackle non-uniform knowledge possessed by different workers.

\subsection{Task Allocation and Team Formation}

The task allocation problem is related to the job scheduling problem which aims at minimizing the load of the machines that have maximal work load. An extended version of the job scheduling problem, in which each job needs to be performed on a set of machines, was proved to be NP-hard~\cite{azar1998line}.

The authors of~\cite{de2007distributed} studied a new variant of the task allocation problem in which the workers are connected in a social network. The workers are assumed to only have local knowledge about resources and hence each task can only be assigned to its neighboring workers. The authors proved this problem to be NP-hard and proposed a max-flow network model to solve it.

Team formation in a non-crowdsourcing environment (e.g. social networks) was studied by \cite{balog2007determining,lappas2009finding}, assuming workers need to communicate when performing a task and communication is costly. The selected team has the skills to complete a task and has minimum communication cost. A different metric, the workload of workers, was studied by \cite{anagnostopoulos2010power} when selecting a team to complete a given task. The authors of~\cite{anagnostopoulos2012online} considered both workload and communication costs when selecting a team to perform a task.

In this paper, we consider pricing mechanisms for the team formation problem, which bridges the gap between budget-feasible mechanisms and the traditional team formation problem. We do not consider workload or communication costs but rather focus on truthfulness which is more important in crowdsourcing environments. In addition, our problem is more general than prior work on budget-feasible mechanisms where a task is always assigned to a single worker.

\section{Model}
\subsection{Requester and Workers}
In our model, a single \emph{requester} posts her task to a crowdsourcing platform. The task has a value $v$ which is the requester's revenue if the task is completed. There is a set of $n$ available \emph{workers} $W=\{w_1, ..., w_n\}$, and the task needs a subset of workers, $S\subseteq W$, to collaborate. When a worker signs up to participate in the task, she should report to the requester what skills she has and how much she expects to be paid. Then the requester selects a set of workers and decides on the payment for each of them.

\textbf{Cost and bid of a worker.} We assume that each worker's cost of doing the task is private information. Each worker $w_i$ has a non-negative cost $c_i\in\mathcal{R}_{\geq0}$ to perform the task, and bids $b_i\in\mathcal{R}_{\geq0}$ when she signs up to do the task. 

We assume that a worker cannot lie about the skills that she has. Some existing crowdsourcing platforms such as MTurk~\cite{mturk} provided qualification test to ensure the validity of workers' skills. A platform can also use work history information to verify one's skill.

\textbf{Utility of the requester.} The requester's utility $U_R$ is the revenue obtained from the completed task, subtracting the payment to the selected workers.
\begin{eqnarray}\label{equ:ur}
U_R =
\begin{cases}
v -\sum_{w_i\in S}p_i & \text{if task is completed} \\
0   & \text{otherwise} \\
\end{cases}
\end{eqnarray}
where $S$ is the set of selected workers and $p_i$ is the payment to worker $w_i$.

\textbf{Utility of a worker.} Worker $w_i$'s utility $u_i$ is the payment she receives subtracting her cost of performing the task.
\begin{eqnarray}
u_i =
\begin{cases}
p_i-c_i & \text{worker $w_i$ is selected} \\
0   & \text{otherwise} \\
\end{cases}
\end{eqnarray}
Both the requester and the workers aim to maximize their respective utilities.

\subsection{Skill Profiles}

The skill profile of the given task is an $l$-dimension vector $s_*=(s_*[1], ..., s_*[l])$, where $s_*[i]=\{0,1\}$ represents that the $i_{th}$ skill is required $(1)$ or not $(0)$. We assume that a maximum of $l$ skills are required for any task. The skill profile of a worker $w_i$ is also an $l$-dimension vector $s_{w_i}=(s_{w_i}[1], ..., s_{w_i}[l])$, similarly defined but representing what skills $w_i$ possesses.

The skill profile of a team $s_T=(s_T[1], ..., s_T[l])$ is defined by a logical OR of the skill profiles of all the individual workers in the team $T$:
\begin{equation}
s_T[j] = \bigvee_{w_i\in T}s_{w_i}[j],\ j=1,...,l
\end{equation}
The team that can complete the task is the team that has all the required skills of the task, i.e.,
\begin{equation}\label{equation:cover-binary}
s_T[j]\geq s_*[j], \ j=1,...,l
\end{equation}

\subsection{Desirable Properties}
\textbf{Computational Efficiency.} The task allocation and pricing mechanism can be executed in polynomial time.

\textbf{Individual Rationality.} No worker is worse off if she is selected to do the task. In other words, each selected worker receives a payment no less than her true cost.

\textbf{Profitability.} The utility of the requester is non-negative.

\textbf{Truthfulness (Incentive Compatibility).}
Bidding her true cost is each worker's dominant strategy. Formally, if $u_i$ and $u_i'$ are the utilities of worker $w_i$ when bidding truthfully and untruthfully, respectively, then a truthful mechanism guarantees that $u_i\geq u_i'$ regardless of what other workers bid.
\begin{theorem}\label{theorem:truthful}
An auction mechanism is truthful if and only if~\cite{singer2010budget}:
\begin{itemize}
\item The allocation rule is monotone: If worker $w_i$ wins the auction by bidding $b_i$, she also wins by bidding $b_i'\leq b_i$.
\item Each winner is paid the threshold price: Worker $w_i$ will not win the auction if she bids higher than this price.
\end{itemize}
\end{theorem}
\subsection{Design Objectives}
We want to design a task allocation and pricing mechanism for the requester to select a team to complete a task she posts, with the objective of minimizing the total payment to the selected workers, i.e.,
\begin{equation}
S = \arg\min_T \sum_{w_i\in T}p_i
\end{equation}
\begin{equation}
s.t.\  s_T[j]\geq s_*[j],~~j=1,...,l
\end{equation}
In addition, the mechanism should satisfy computational efficiency, individual rationality, profitability and truthfulness.

\section{Mechanisms}

%
%
%
%
%
%
%
%
%
%
%
%
%

\subsection{Optimal Mechanism}
This mechanism selects the team with the lowest total bid that is able to complete the task, by taking a brute-force approach to attempt all the possible $2^n-1$ teams (excluding the empty set) of the $n$ workers. We refer to this mechanism as \textsf{OPT}.

\begin{lemma}\label{OPTlem}
\textsf{OPT} is individually rational, profitable, but not computationally efficient or truthful. Time complexity of \textsf{OPT} is $O(2^n)$.
\end{lemma}

The proof is deferred to \cite{techrep} due to space constraint, as is also the case for Lemma~\ref{GREEDYlem} and Theorem~\ref{thm:truteam}.

\subsection{Greedy Mechanism}
This heuristic mechanism selects the worker with the minimum cost per \textit{marginal skill contribution} until a team that can complete the task is formed or all the workers have been considered. It pays each selected worker her bid.

In the above, $w_i$'s \emph{marginal skill contribution}, $\Delta_i(S)$, is defined with respect to an existing worker set $S$ as the number of uncovered skills that $w_i$ can cover if selected into $S$:
\begin{equation}
\Delta_i(S)=s_{S\cup \{w_i\}}\cdot s_*-s_S\cdot s_*
\end{equation}
where $s_S$ is the skill profile of team $S$, and $s_S\cdot s_*$ is the inner product of vectors $s_S$ and $s_*$. In each iteration, the \textsf{GREEDY} mechanism always selects the worker who has the lowest cost per marginal skill contribution, $\frac{b_i}{\Delta_i(S)}$.

\begin{lemma}\label{GREEDYlem}
\textsf{GREEDY} is computationally efficient, individually rational, profitable, but not truthful. Time complexity of \textsf{GREEDY} is $O(n^2)$.
\end{lemma}

\subsection{VCG-based Mechanism}
We adapt the traditional Vickrey-Clarke-Groves (VCG) mechanism~\cite{nisan2007algorithmic} to our problem, and design a mechanism called \textsf{VCG}.

\textbf{Allocation rule.} \textsf{VCG} selects the team with the lowest total cost, i.e.,
\begin{eqnarray}\label{equ:VCG}
S=\arg\min_{T}\sum_{w_i\in T}c_i
\end{eqnarray}
\[s.t.\  s_S[j]\geq s_*[j],~~j=1,...,l\]

\textbf{Payment rule.} \textsf{VCG} pays each selected worker $w_i$ the difference between the optimal welfare (for the other workers) if $w_i$ was not participating and welfare of the other workers with respect to the selected team:
\begin{eqnarray}
p_i=(\min_{T}\sum_{w_j\in T \wedge w_i\notin T}c_j) - \sum_{w_j\in S \wedge j\neq i}c_j
\end{eqnarray}
\[s.t.\  s_T[j]\geq s_*[j],~~j=1,...,l\]
where $S$ is defined in (\ref{equ:VCG}).

\begin{lemma}
\textsf{VCG} is individually rational, profitable, truthful, but not computationally efficient. Time complexity of \textsf{VCG} is $O(n2^n)$.
\end{lemma}

\begin{proof}
Suppose $u_i,u_i'$, $p_i,p_i'$ are the corresponding utilities and payments of $w_i$, when she bids truthfully and untruthfully, respectively. Let $S,S'$ denote the selected team when $w_i$ bids truthfully and untruthfully, respectively.
\begin{equation}
\begin{split}
u_i'
&=p_i'-c_i\\
&=(\min_{T}\sum_{w_j\in T \wedge w_i\notin T}c_j) - \sum_{w_j\in S' \wedge j\neq i}c_j -c_i\\
&=(\min_{T}\sum_{w_j\in T \wedge w_i\notin T}c_j) - \sum_{w_j\in S' }c_j\\
&\leq(\min_{T}\sum_{w_j\in T \wedge w_i\notin T}c_j) - \sum_{w_j\in S }c_j \\
&= u_i
\end{split}
\end{equation}
which proves the truthfulness. The proof of other properties are deferred to \cite{techrep} due to space constraint.
\end{proof}
\subsection{Efficient and Truthful Mechanism (\textsf{TruTeam})}
\textsf{OPT} is an optimal allocation mechanism if every worker bids truthfully. \textsf{GREEDY} is computationally efficient but not a truthful mechanism. \textsf{VCG} is a truthful mechanism but is not computationally efficient.

In this section, we present the mechanism \textsf{TruTeam} (Mechanism~\ref{TruTeamAlgo}) which satisfies all the four properties (i.e., \textsf{TruTeam} is computationally efficient, individually rational, profitable, and truthful).

\IncMargin{1em}
\begin{algorithm}\label{TruTeamAlgo}

\renewcommand{\algorithmcfname}{Mechanism}
\caption{Efficient and Truthful mechanism for Team formation (\textsf{TruTeam})\label{alg:TruTeam}}
\KwIn{$b_{1\sim n}$, $s_{w_1\sim w_n}$, $v$, $s_*$}
\KwOut{$S$,  $p_{1\sim n}$ }

$S \leftarrow \emptyset$;
$p_{1\sim n}\leftarrow 0$;


\Repeat{$s_{S}$ cover $s_*$ or $W\backslash S=\emptyset$}{
    $w_i \leftarrow \arg\min_{w_i\in W\backslash S }\frac{b_i}{\Delta_i(S)}$;\\
    $W' \leftarrow W \backslash \{S \cup \{w_i\}\}$;\\
    $T \leftarrow S$;\\
    \Repeat{$\Delta_i(T)=0$ or $v-p_i <0$ }{
        $w_j \leftarrow \arg\min_{w_j \in W' \backslash T}\frac{b_j}{\Delta_j(T)}$;\\
        $p_i \leftarrow \max\{\frac{b_j}{\Delta_j(T)}\times \Delta_i(T), p_i\}$;\\
        $T \leftarrow T\cup\{j\}$;\\
    }
    \If {$v\geq p_i$}{
        $S\leftarrow S \cup \{w_i\}$;\\
        $v\leftarrow v - p_i$;\\
    }
    \Else {
        $W\leftarrow W \backslash \{w_i\}$;\\
    }
}
\If{$s_S$ not cover $s_*$}{
    $S \leftarrow \emptyset $;
    $p_{1\sim n}\leftarrow 0$
}

\Return($S$; $p_1,...,p_n$)
\end{algorithm}
\DecMargin{1em}

\textbf{Allocation rule.}
In each iteration, it selects the worker who has the smallest cost per marginal skill contribution, i.e., the lowest $\frac{b_i}{\Delta_i(S)}$. This is the same as \textsf{GREEDY}.

\textbf{Payment rule.} The intuition of the payment rule is to pay each selected worker the highest cost she can report while still being selected~\cite{goel2014allocating}. This is the ``threshold price" stated in Theorem~\ref{theorem:truthful}, and we will show that overbidding under \textsf{TruTeam} does no good to improve a worker's utility.

Now, we explain in detail how to determine the payment to each selected worker. When computing the payment to worker $w_i$, let's see how this mechanism selects a team without $w_i$'s participation. It selects from set $W\backslash \{{T\cup \{w_i\}}\}$ ($T$ is the selected worker set before $w_i$) the worker ($w_{j_1}$) who minimizes the value $\frac{b_j}{\Delta_j(S)}$.
Therefore
\begin{equation}
b_{i}\leq \frac{b_{j_1}}{\Delta_{j_1}(T)}\times {\Delta_i(T)}
\end{equation}
Otherwise, if $\frac{b_{i}}{\Delta_{i}(T)} > \frac{b_{j_1}}{\Delta_{j_1}(T)}$, we would have selected $w_{j_1}$ instead of $w_i$ according to the allocation rule of this mechanism. Therefore, we set the payment to worker $w_i$ equal to this value:
\begin{equation}
p_i'=\frac{b_{j_1}}{\Delta_{j_1}(T)}\times {\Delta_i(T)}
\end{equation}
However, $p_i'$ may not be the highest bid that $w_i$ can report while still being able to be selected, because $T\cup \{w_{j_1}\}$ may not cover all the skills required by the task. Suppose $T\cup \{w_{j_1},w_{j_2},...,w_{j_k}\}$ (without $w_i$) is the set of workers selected according to this mechanism.
We set the payment equal to the following threshold price as described in lines 6-10.
\begin{equation}\label{equ:payment}
p_i=\max_{j\in \{j_1,j_2,...,j_k\}} \{  \frac{b_{j}}{\Delta_{j}(T)}\times {\Delta_i(T)}  \}
\end{equation}
Note that, $T$ is updated every time by including a new worker $w_{j_x}$ ($x=1,2,...,k$), i.e., $T=T\cup \{w_{j_x}\}$.

In order to be profitable, $w_i$ is selected to perform this task only if the task's remaining value is not less than $p_i$ (we also update the set of selected workers as $S=S\cup\{w_i\}$), as shown in lines 11-13. Otherwise we skip $w_i$ and consider next candidate worker as described in lines 14-15.

Repeat the above process until the task can be completed by the set of workers $S$ or all the workers have been considered.

\begin{lemma}
\textsf{TruTeam} is computationally efficient, with a time complexity of $O(n^2l)$.
\end{lemma}

\begin{proof}
Selecting the worker who has the minimal value $\frac{b_i}{\Delta_i(S)}$ takes $O(n)$ time. Deciding the payment for the selected worker takes $O(nl)$. Since there are $n$ workers, time complexity of this mechanism is $O(n^2l)$.
\end{proof}

\begin{lemma}
\textsf{TruTeam} is individually rational.
\end{lemma}

\begin{proof}
From the above payment rule, we can see that for any selected worker $w_i$.
 \begin{equation}
\begin{split}
b_i
&\leq \frac{b_{j_1}}{\Delta_{j_1}(S)}\times {\Delta_i(S)}\\
&\leq max_{j\in \{j_1,j_2,...,j_k\}} \{  \frac{b_{j}}{\Delta_{j}(T)}\times {\Delta_i(T)}  \}\\
&= p_i
\end{split}
\end{equation}
We assume a worker will not bid bellow her true cost, i.e. $b_i\geq c_i$ (In fact, we will show in Lemma~\ref{lem:TruTeam} that $b_i=c_i$ ). Therefore, $w_i$'s utility is $u_i=p_i-c_i\geq p_i-b_i\geq0$.
\end{proof}

\begin{lemma}
\textsf{TruTeam} is profitable.
\end{lemma}

\begin{proof}
Every time when a worker is considered, we check if the remaining value of the task can cover the payment to this worker. If not, the worker will not be selected. This process guarantees that the value of this task is more than the total payment to the whole team.
\end{proof}

\begin{lemma}\label{lem:TruTeam}
\textsf{TruTeam} is truthful.
\end{lemma}

\begin{proof}
According to Theorem~\ref{theorem:truthful}, we need to prove that (1) the allocation rule is monotone and (2) the payment to each selected worker $w_i$ is the threshold price $p_i$.

The monotonicity of the allocation rule is obvious, since if $w_i$ is selected, she will also be selected by bidding a smaller value which leads to a smaller cost per marginal skill contribution.

For the threshold price, recall $p_i$ from equation (\ref{equ:payment}).
If $b_i>p_i$, $w_i$ will be placed after the last selected worker,
thus she will not be selected to perform the task. Therefore, $p_i$ is the threshold price.\footnote{Bidders will not underbid, either. The intuition is that if $w_i$ underbids and is selected, her payment will not cover her cost. See \cite{techrep} for the details.}
\end{proof}

\begin{theorem}\label{thm:truteam}
\textsf{TruTeam} is computationally efficient, individually rational, profitable and truthful.
\end{theorem}

\section{Evaluation}
In this section, we compare the performance of the four mechanisms, namely \textsf{OPT}, \textsf{VCG}, \textsf{GREEDY} and \textsf{TruTeam} in terms of the following metrics:
\begin{itemize}
\item \textbf{Running Time}: the actual CPU time on a computer.
\item \textbf{Requester's Utility}: defined in Eqn. (\ref{equ:ur}).
\item \textbf{Truthfulness}: We verify the truthfulness of \textsf{TruTeam} by evaluating workers' utility if they overbid or underbid.
\end{itemize}

\subsection{Simulation Setup}
We generate two different datasets to evaluate our mechanisms. A \textsf{Small} dataset is used to evaluate all the four mechanisms and a \textsf{Large} dataset is used to evaluate the two computationally efficient mechanisms (i.e., \textsf{GREEDY} and \textsf{TruTeam}). The parameters of the two datasets are listed in Table \ref{table:ex-setting} and explained bellow.

We set the value of the task $v=500$ which is unknown to workers. Each worker's true cost $c_i$ is uniformly drawn from $[C_1, C_2]$. In the case that all workers are truthful, we set $b_i=c_i$. In the case of overbidding, we randomly select $k\in[1,n]$ workers and let each of them overbid a random value $r_i$ (i.e., $b_i=c_i+r_i$) where $r_i\in [1,v]$. We do not consider underbidding, since no rational worker will underbid under these four mechanisms.

To generate each worker's skill profile, we first generate the number of skills she has, using the normal distribution ($u,\sigma$). Suppose $w_i$ has $x$ skills, then we randomly assign $x$ different skills out of all the $l$ skills to her.
\begin{table}
\centering
\caption{Parameters of our two datasets}\label{table:ex-setting}
{\footnotesize 
\begin{tabular}{|c|p{2.0cm}|p{1.9cm}|p{1.0cm}|c|} \hline
        &$n$: no. of workers & $l$: no. of skills & ($u,\sigma$) & $[C_1,C_2]$\\ \hline

\textsf{Small} &  $n=10\sim 25$, fixing $l=5$ & $l=1\sim 10$, fixing $n=20$   &   ($l/3,0.4$) &   $[1,v/5]$             \\ \hline \

\textsf{Large} &  $n=10\sim 3000$, fixing $l=50$ & $l=1\sim100$, fixing $n=1000$  &   ($l/5,0.4$) &   $[1,v]$            \\ \hline

\end{tabular}
}
\end{table}

All the simulations were run on a Windows PC with a 3.40GHz CPU and 8 GB memory. Each data point is averaged over 100 measurements.

\subsection{Results}
Fig.~\ref{fig:Smalldataset} shows the comparison of the four mechanisms conducted over the \textsf{Small} dataset. We observe in Fig.~\ref{fig:Smalldataset} (a,b,c,d) that \textsf{OPT} and \textsf{VCG} do not scale well when the number of workers or skills becomes large. However, all the four mechanisms achieve strictly positive requester's utility (e,f,g,h). 
Therefore, in this section, we focus on the computationally efficient mechanisms, \textsf{GREEDY} and \textsf{TruTeam}, and explain in detail the results pertaining to them collected on the \textsf{Large} dataset.

\begin{figure*}
\subfigure[Truthful bidding; $l=5$]{\label{fig:time-small-worker-truthful}\includegraphics[trim=2mm 2mm 8mm 7mm,clip,angle=0,width=0.23\textwidth]{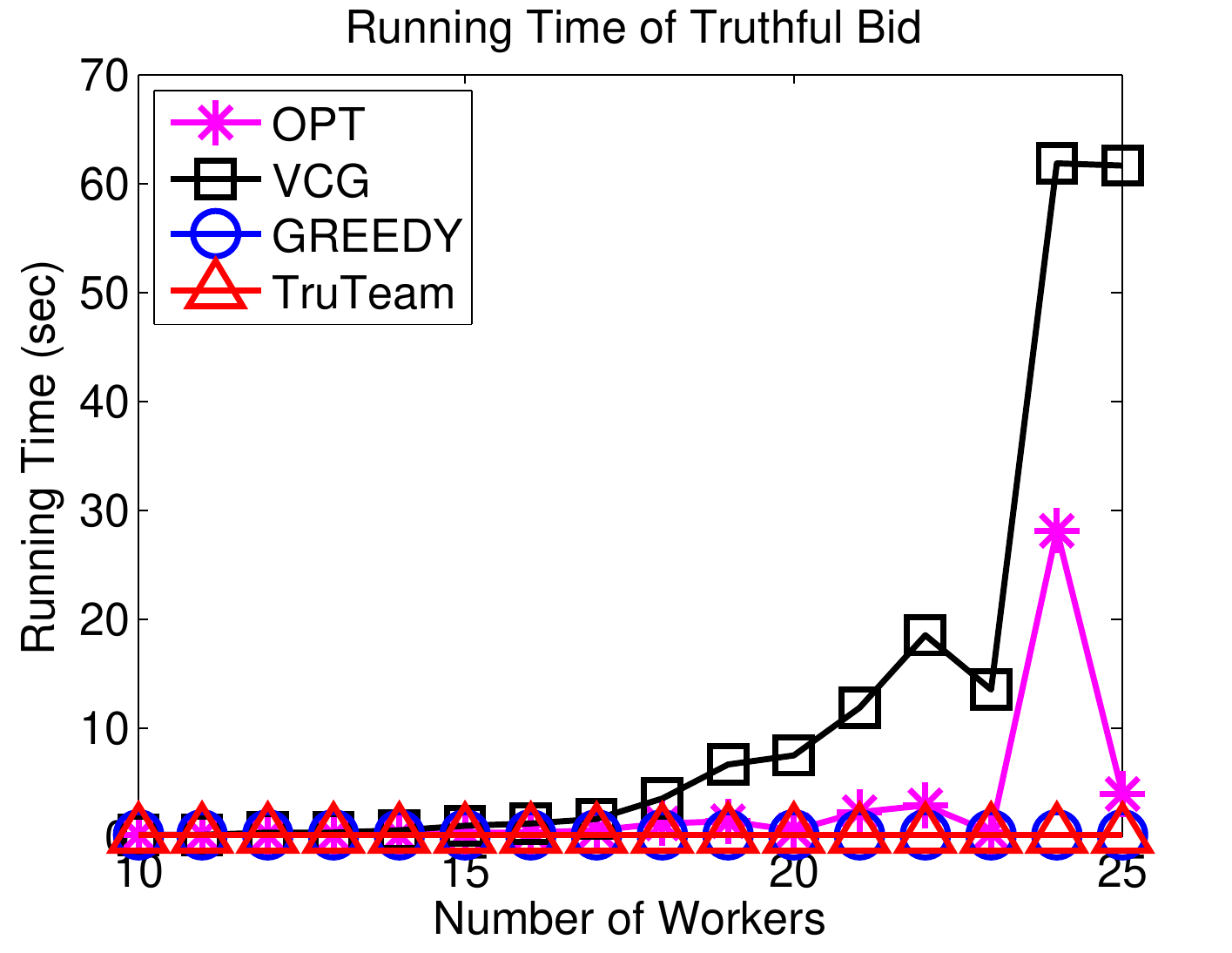}}
\hfill
\subfigure[Overbidding; $l=5$]{\label{fig:time-small-worker-over}\includegraphics[trim=2mm 2mm 8mm 7mm,clip,angle=0,width=0.23\textwidth]{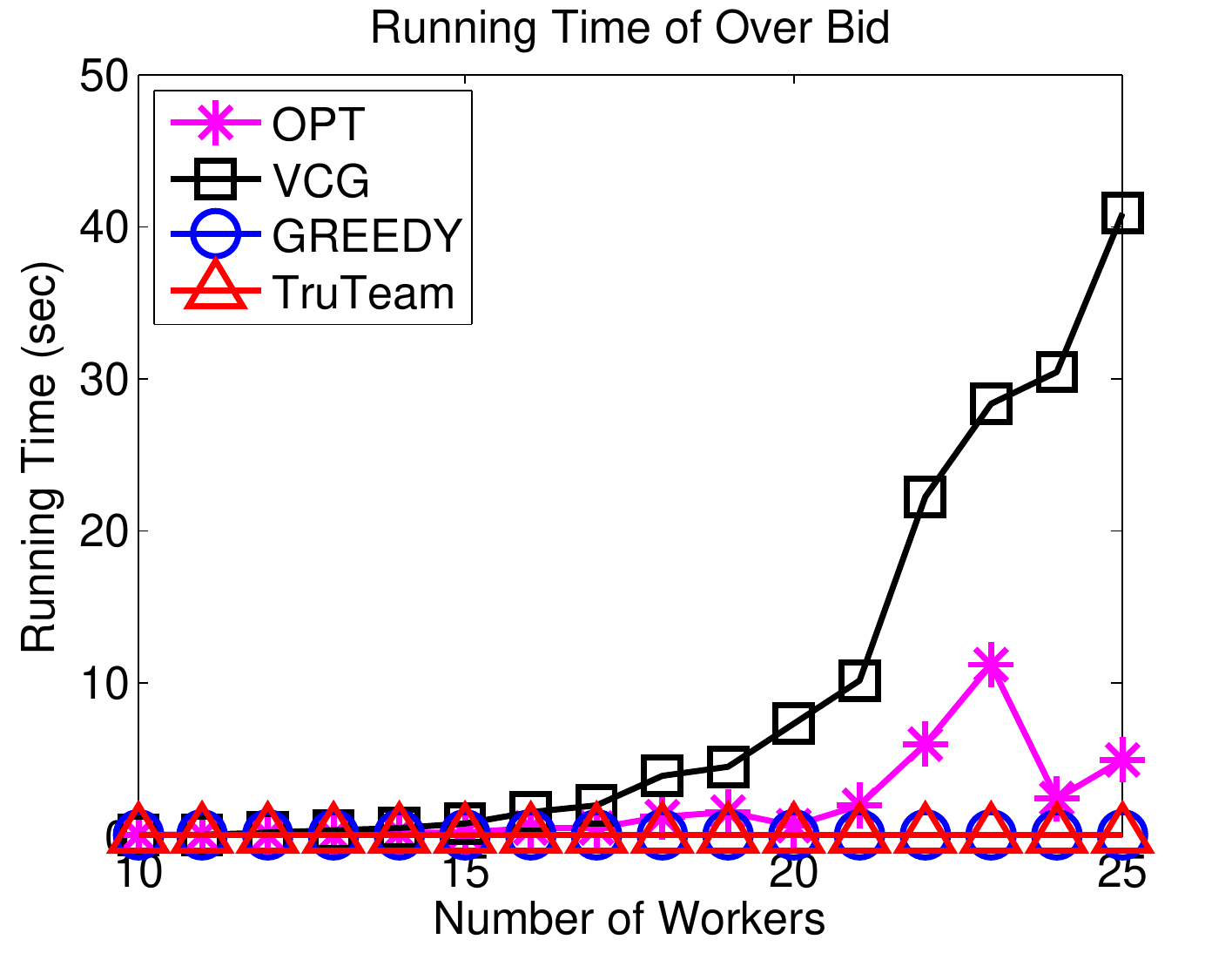}}
\hfill
\subfigure[Truthful bidding; $n=20$]{\label{fig:time-small-skill-truthful}\includegraphics[trim=2mm 2mm 8mm 7mm,clip,angle=0,width=0.23\textwidth]{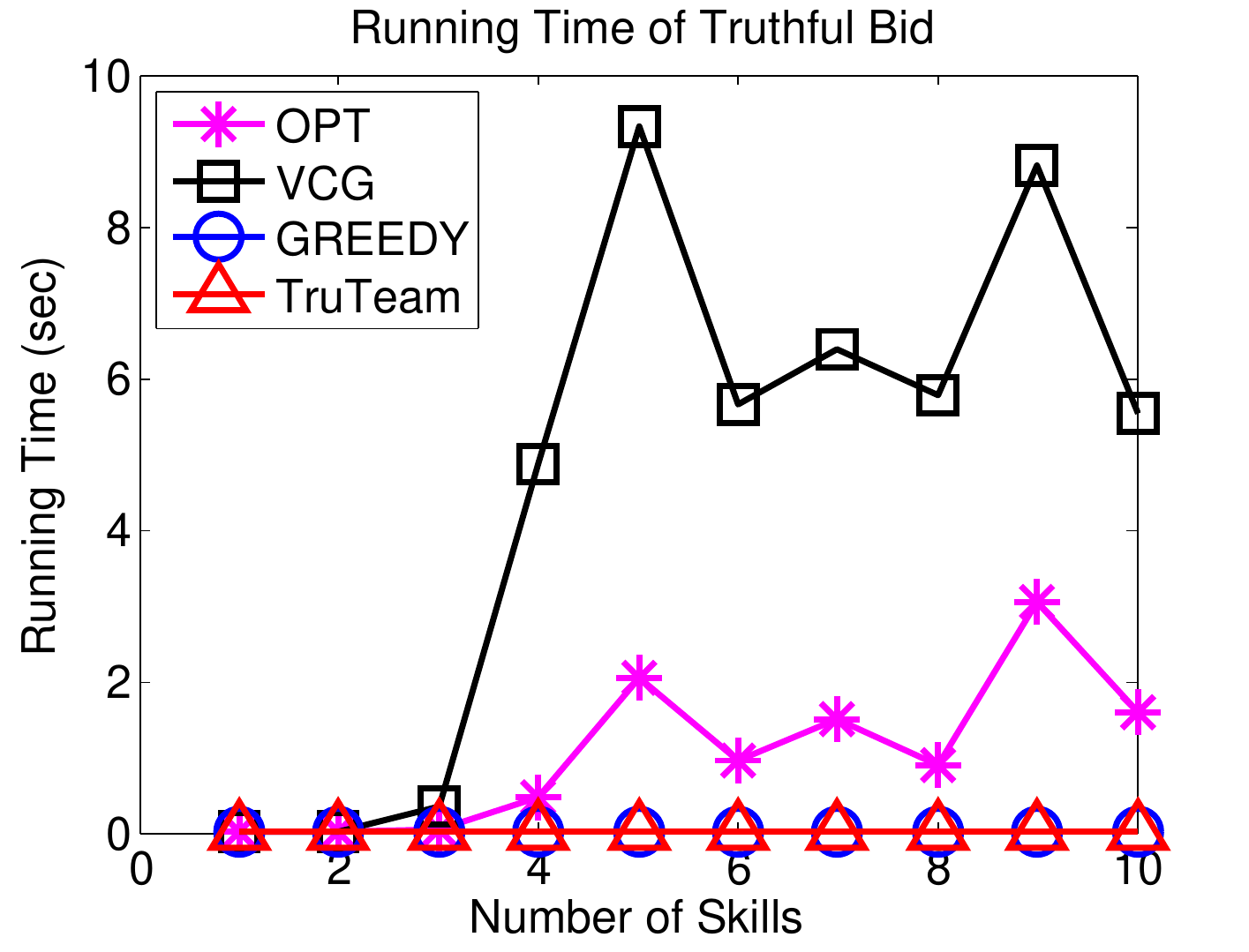}}
\hfill
\subfigure[Overbidding; $n=20$]{\label{fig:time-small-skill-over}\includegraphics[trim=2mm 2mm 8mm 7mm,clip,angle=0,width=0.23\textwidth]{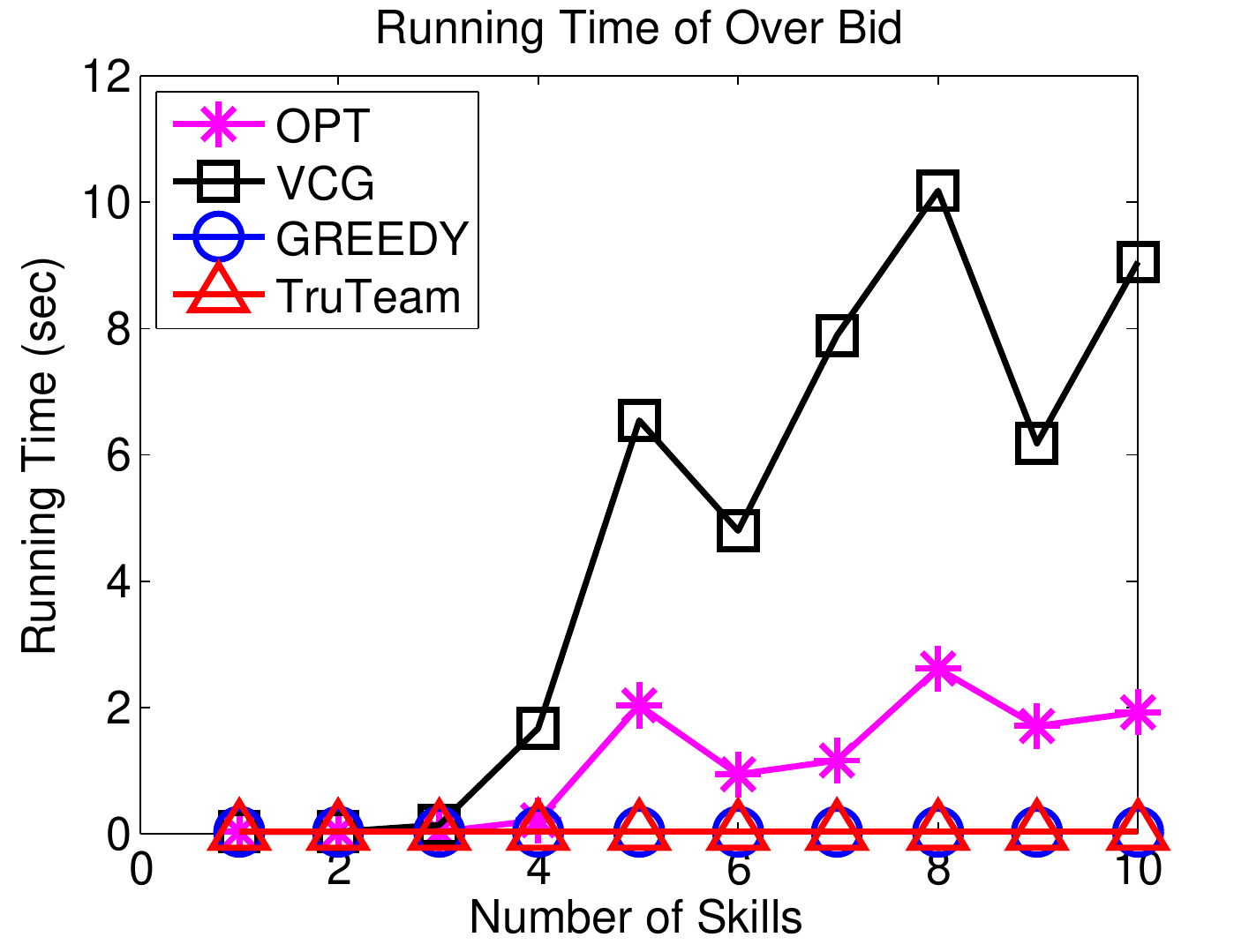}}
\\
\subfigure[Truthful bidding; $l=5$]{\label{fig:utility-small-worker-truthful}\includegraphics[trim=2mm 2mm 8mm 7mm,clip,angle=0,width=0.23\textwidth]{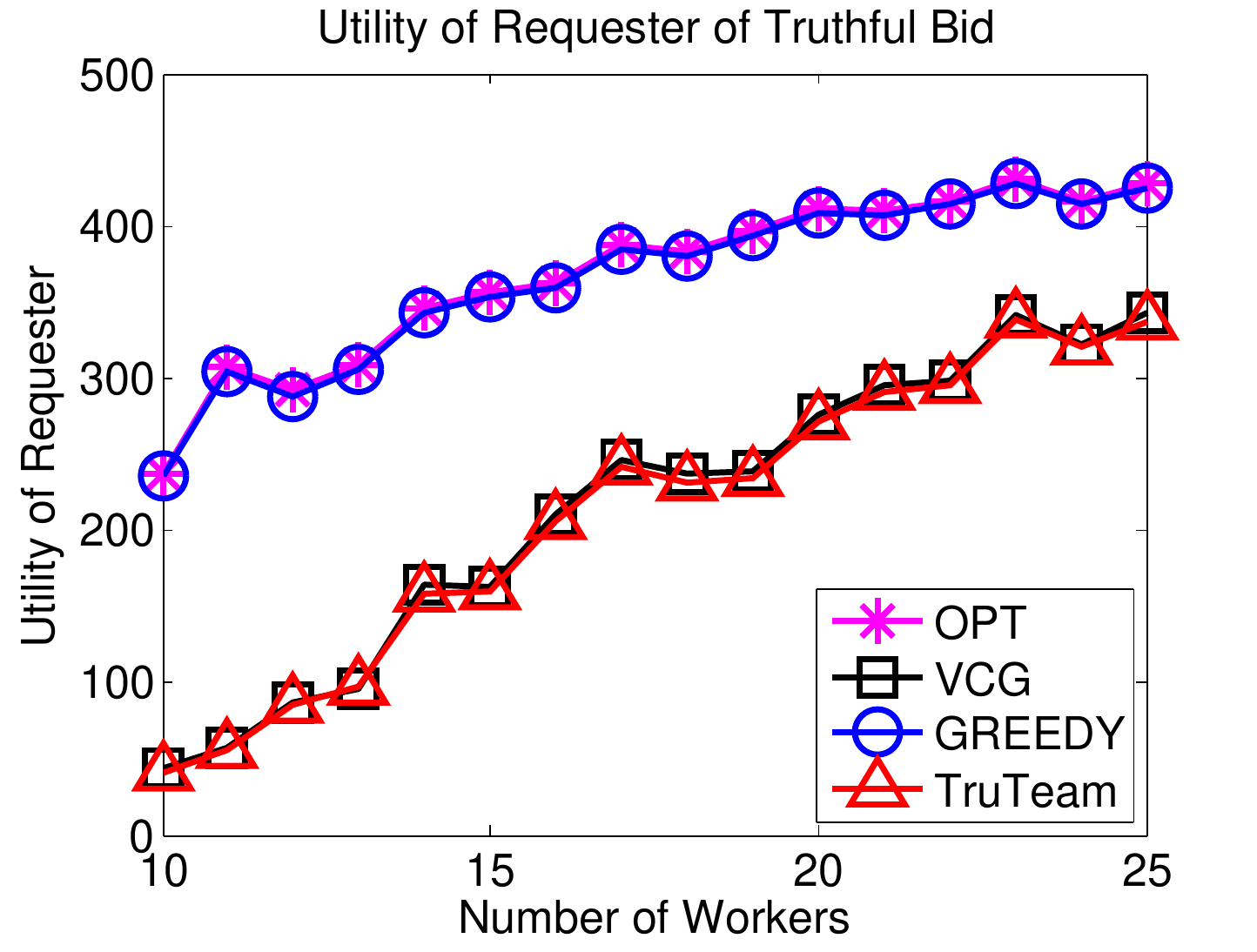}}
\hfill
\subfigure[Overbidding; $l=5$]{\label{fig:utility-small-worker-over}\includegraphics[trim=2mm 2mm 8mm 7mm,clip,angle=0,width=0.23\textwidth]{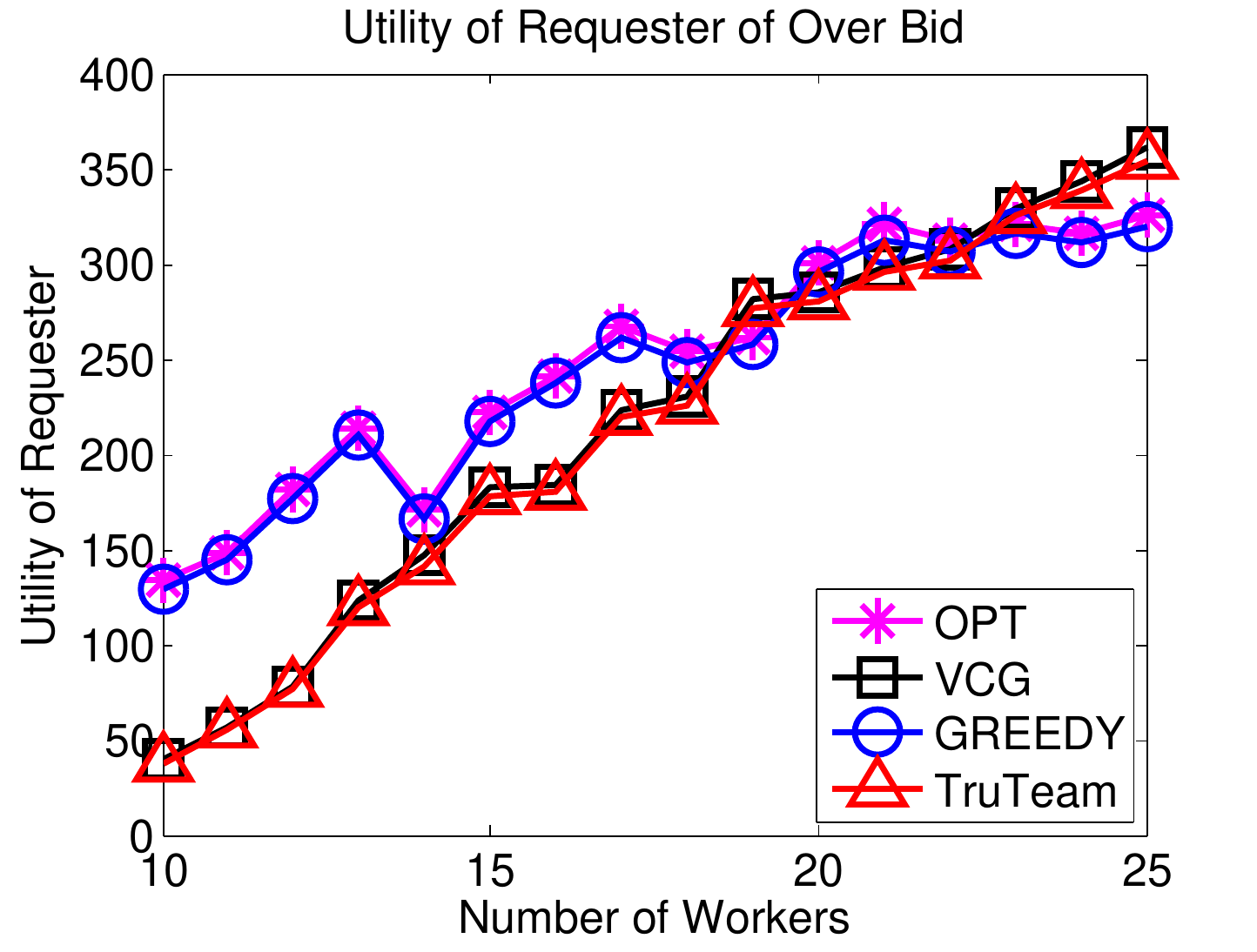}}
\hfill
\subfigure[Truthful bidding; $n=20$]{\label{fig:utility-small-skill-truthful}\includegraphics[trim=2mm 2mm 8mm 7mm,clip,angle=0,width=0.23\textwidth]{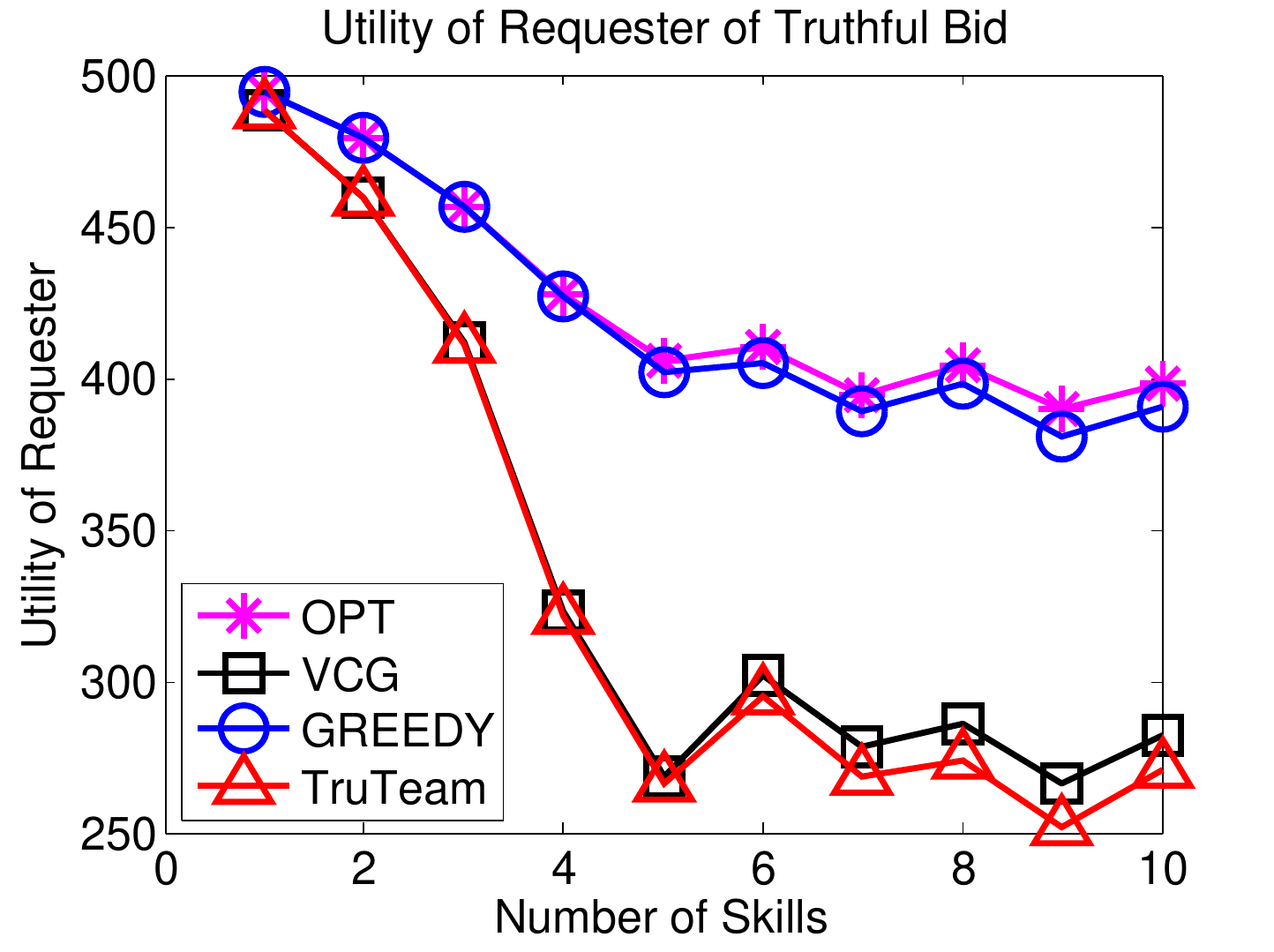}}
\hfill
\subfigure[Overbidding; $n=20$]{\label{fig:utility-small-skill-over}\includegraphics[trim=2mm 2mm 8mm 7mm,clip,angle=0,width=0.23\textwidth]{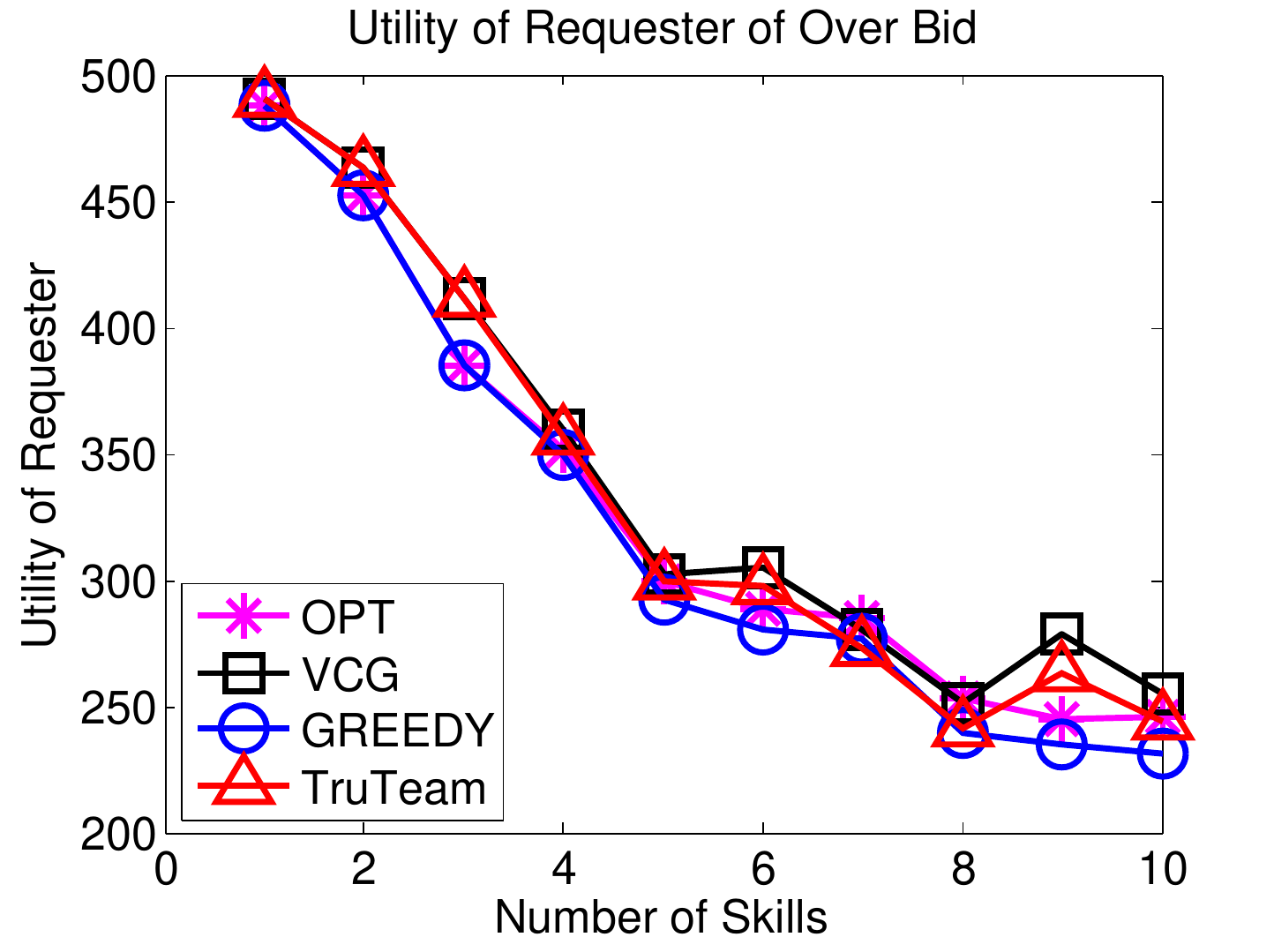}}
\caption{Simulation results on the \textsf{Small} Dataset.} \label{fig:Smalldataset}
\end{figure*}

\subsubsection{Running Time}
The results are presented in Fig.~\ref{fig:Largedataset} (a,b,c,d). 
Generally, as the number of workers or the number of skills increases, running time of both mechanisms increases. Although \textsf{GREEDY} outperforms \textsf{TruTeam}, since the payment determination process of \textsf{TruTeam} is more complicated than that of \textsf{GREEDY}, \textsf{TruTeam} maintains a high efficiency. For example, it completes in about 0.6 second even when the number of worker reaches 3000. More importantly, \textsf{TruTeam} ensures truthfulness which is crucial to incentive mechanisms to counteract possible cheating behaviors in practice.

In Fig.~\ref{fig:time-big-worker-truthful} and \ref{fig:time-big-worker-over}, running time of \textsf{TruTeam} contains a small peak when the number of workers is between 100 and 300. This happens when the task is complex (requiring skills as many as $l=50$) and the number of workers is relatively small ($n<300$). Thus the team size is fairly large and \textsf{TruTeam} needs to check every other worker's bid and skills when deciding the payment to each selected worker, resulting in higher running time.

\subsubsection{Requester's Utility}

Fig. \ref{fig:Largedataset} (e,f,g,h) 
present the requester's utility in different settings.
Generally, the requester's utility increases as the number of worker increases (e,f), since the requester has more ``cheaper'' workers to select from. On the other hand, the requester's utility drops as the number of required skills becomes larger (g,h), which is a natural result of the increase of complexity of the task.

In the case of truthful bidding, as is shown in Fig.~\ref{fig:utility-big-worker-truthful}~\ref{fig:utility-big-skill-truthful}, the untruthful mechanism (\textsf{GREEDY}) yields higher utility than the truthful mechanism (\textsf{TruTeam}) because \textsf{TruTeam} pays each selected worker no less than her bid. 

%

However, in the case of overbidding which is a more realistic setting, Fig.~\ref{fig:utility-big-worker-over} \ref{fig:utility-big-skill-over} show that \textsf{TruTeam} outperforms \textsf{GREEDY}. This demonstrates that in real crowdsourcing markets where workers are strategic and speculating higher payment (e.g., by trying to overbid), \textsf{TruTeam} generates higher profit for the requester by ensuring truthful bidding.

\begin{figure*}
\subfigure[Truthful bidding; $l=50$]{\label{fig:time-big-worker-truthful}\includegraphics[trim=2mm 2mm 8mm 7mm,clip,angle=0,width=0.24\textwidth]{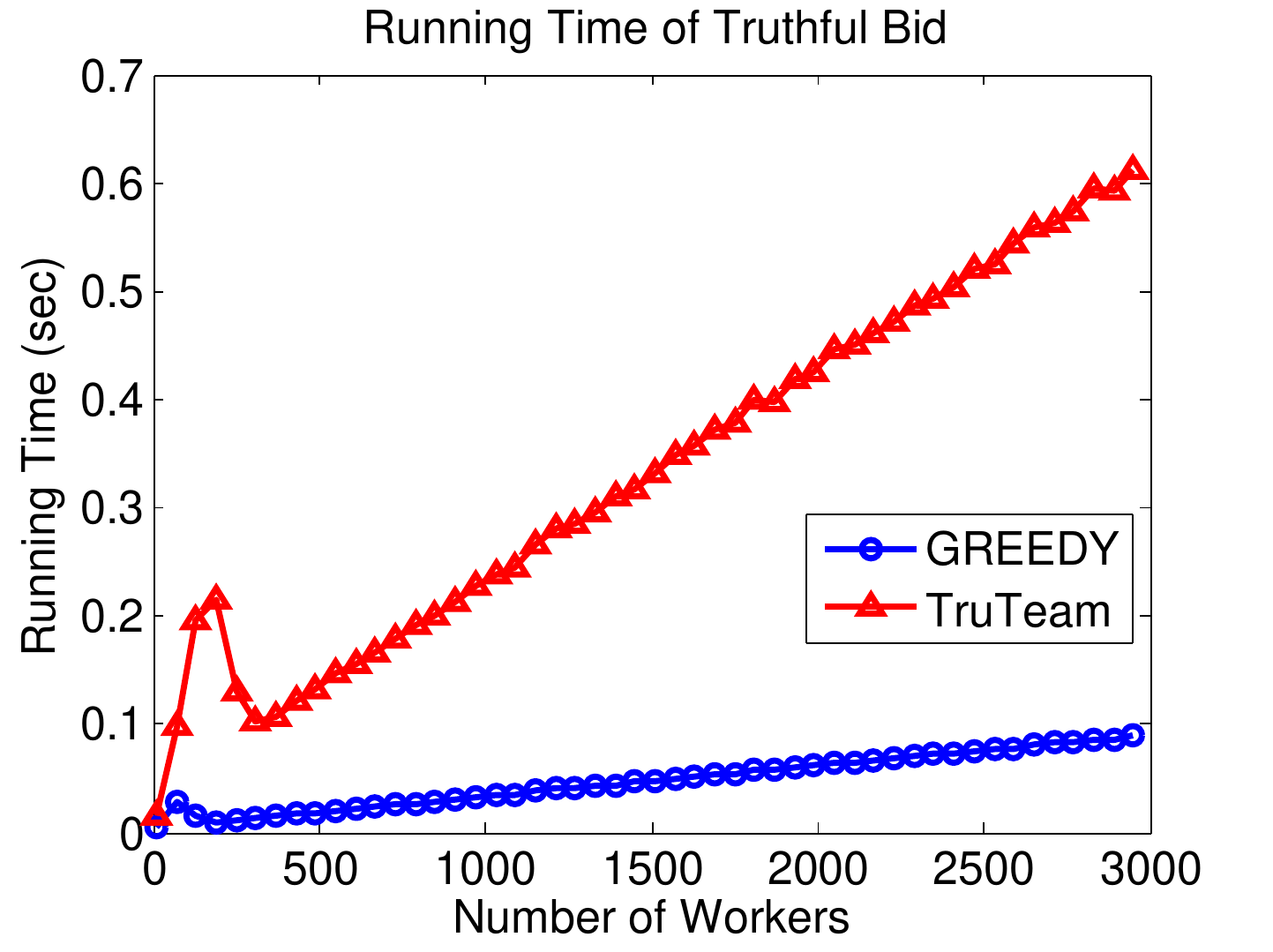}}
\hfill
\subfigure[Overbidding; $l=50$]{\label{fig:time-big-worker-over}\includegraphics[trim=2mm 2mm 8mm 7mm,clip,angle=0,width=0.24\textwidth]{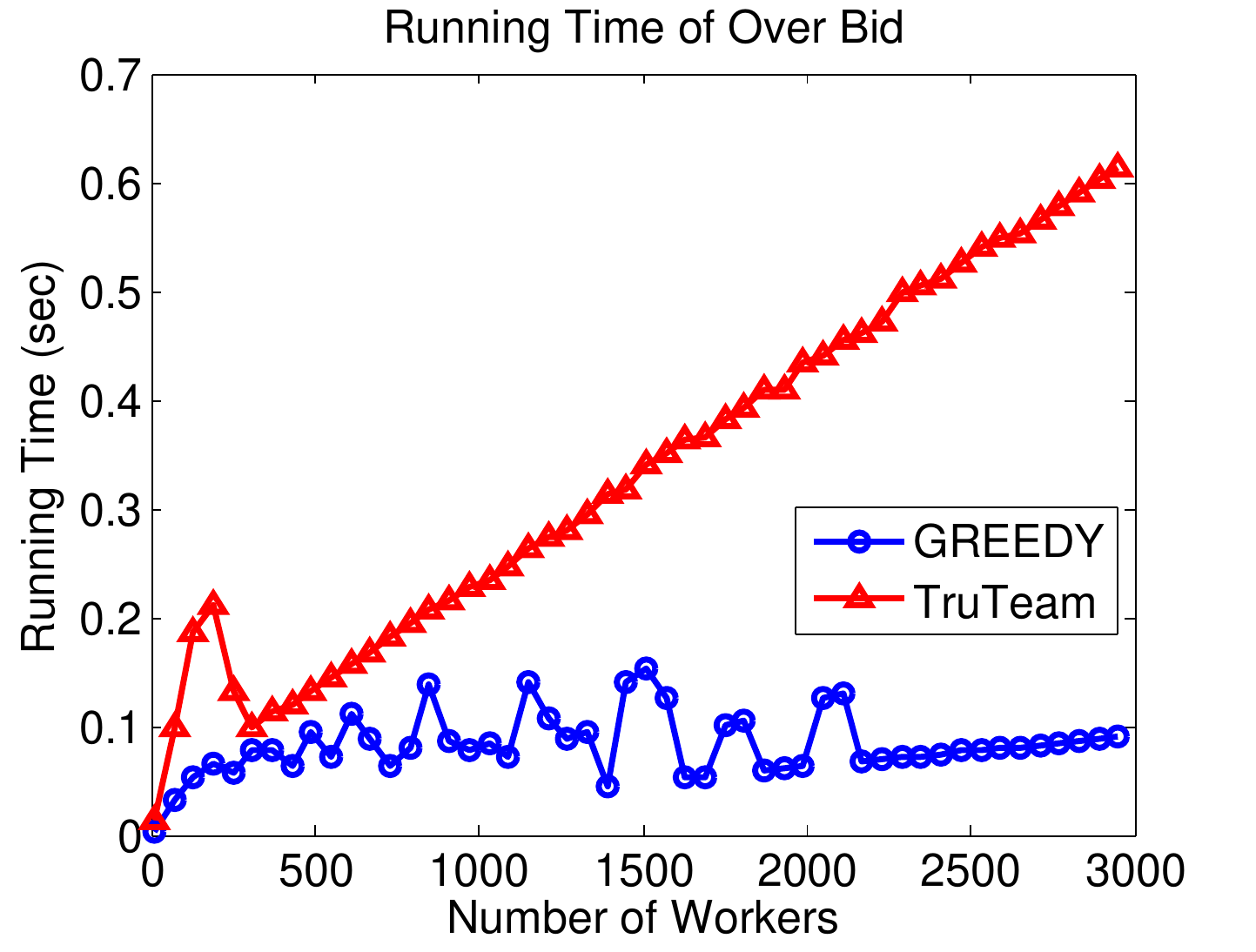}}
\hfill
\subfigure[Truthful bidding; $n=1000$]{\label{fig:time-big-skill-truthful}\includegraphics[trim=2mm 2mm 8mm 7mm,clip,angle=0,width=0.24\textwidth]{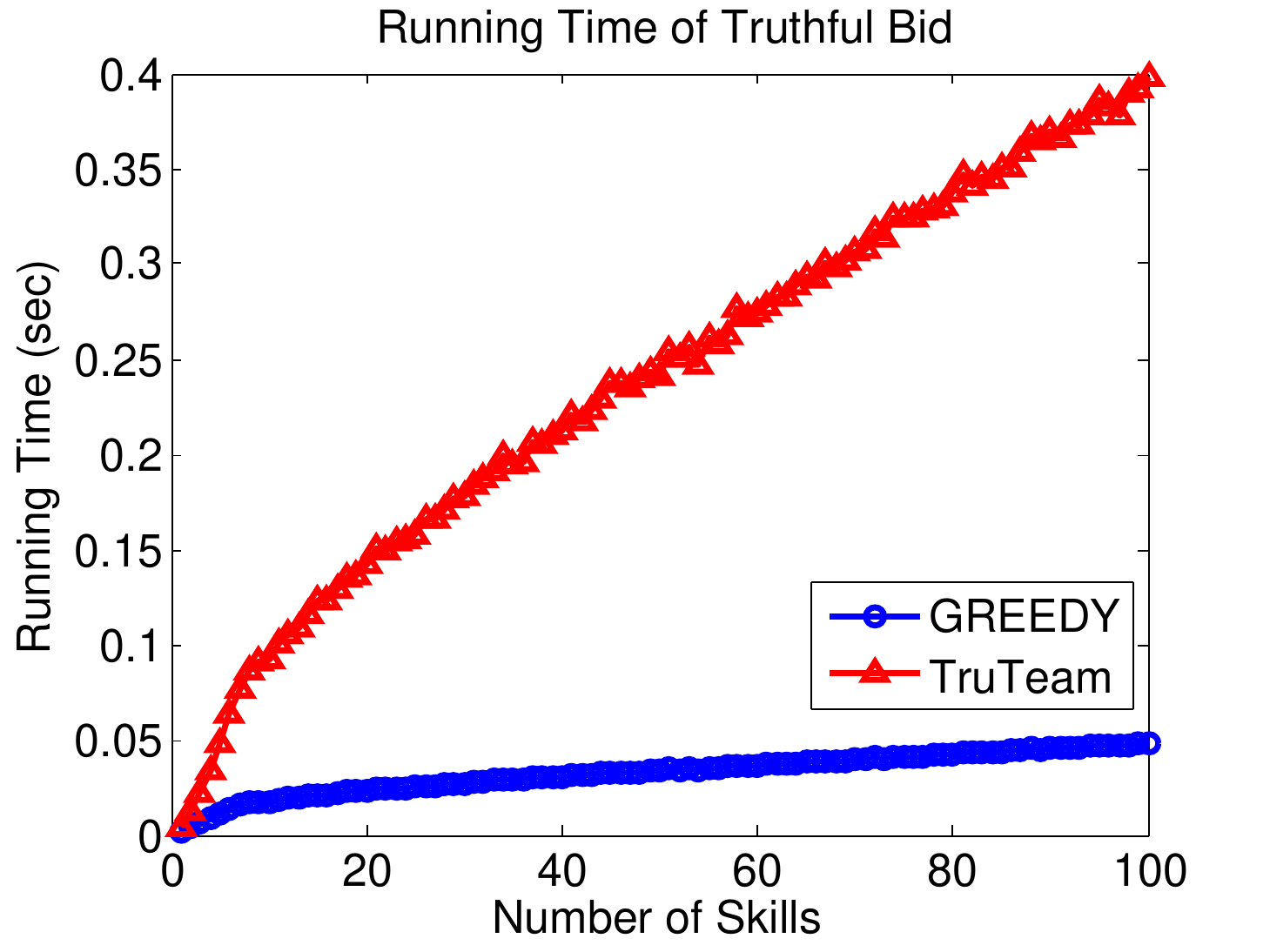}}
\hfill
\subfigure[Overbidding; $n=1000$]{\label{fig:time-big-skill-over}\includegraphics[trim=2mm 2mm 8mm 7mm,clip,angle=0,width=0.24\textwidth]{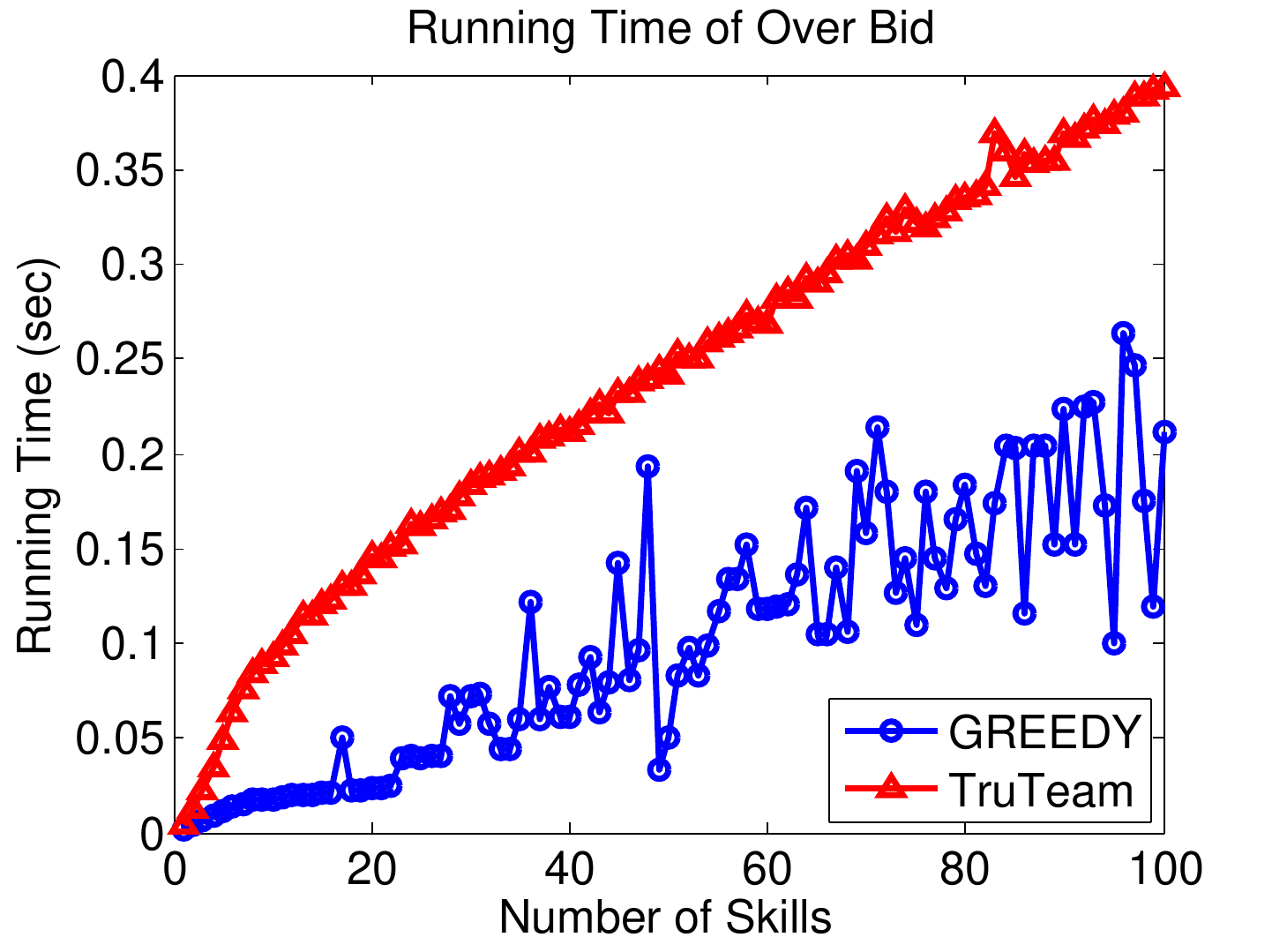}}
\\
\subfigure[Truthful bidding; $l=50$]{\label{fig:utility-big-worker-truthful}\includegraphics[trim=2mm 2mm 8mm 7mm,clip,angle=0,width=0.24\textwidth]{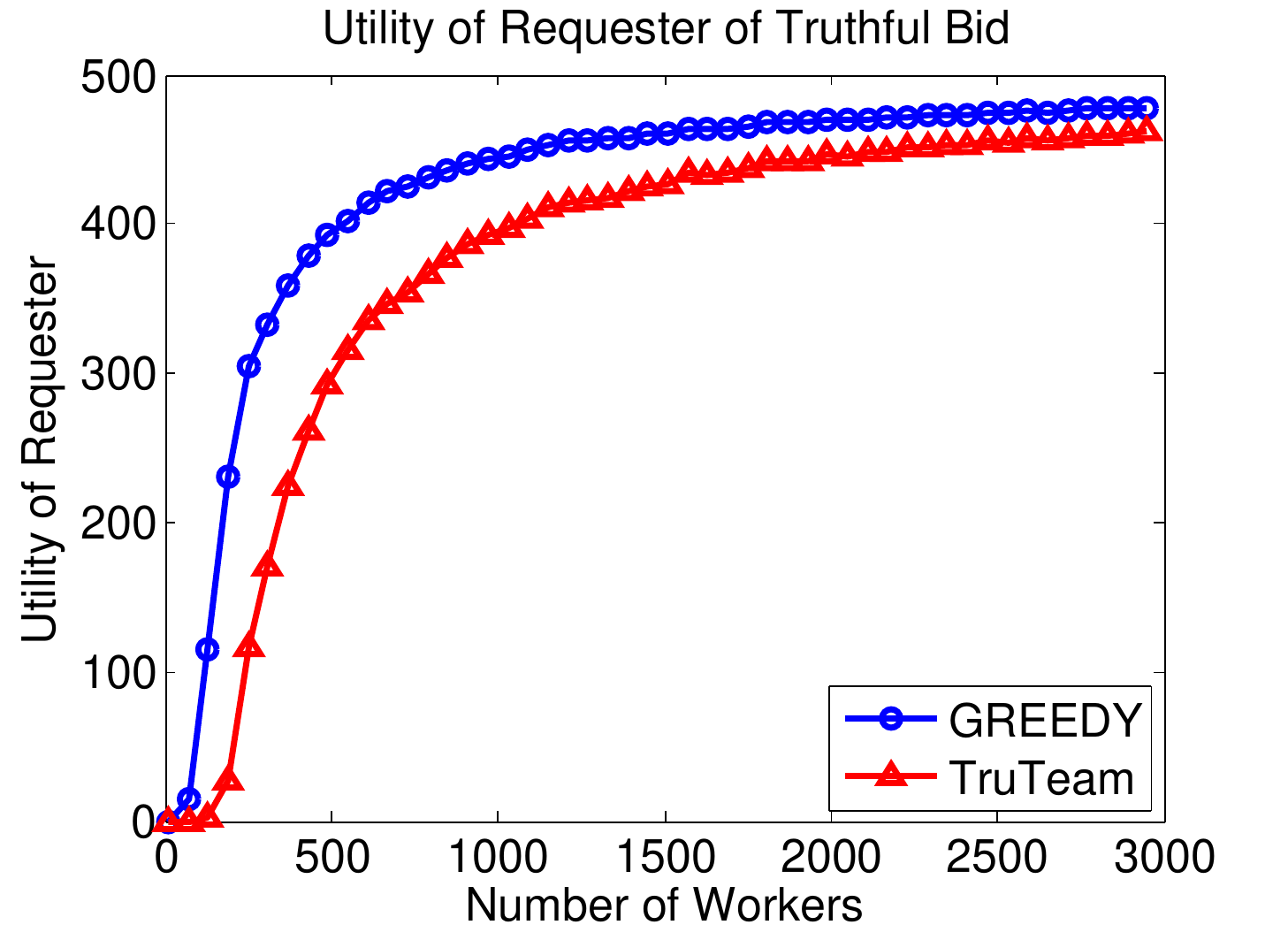}}
\hfill
\subfigure[Overbidding; $l=50$]{\label{fig:utility-big-worker-over}\includegraphics[trim=2mm 2mm 8mm 7mm,clip,angle=0,width=0.24\textwidth]{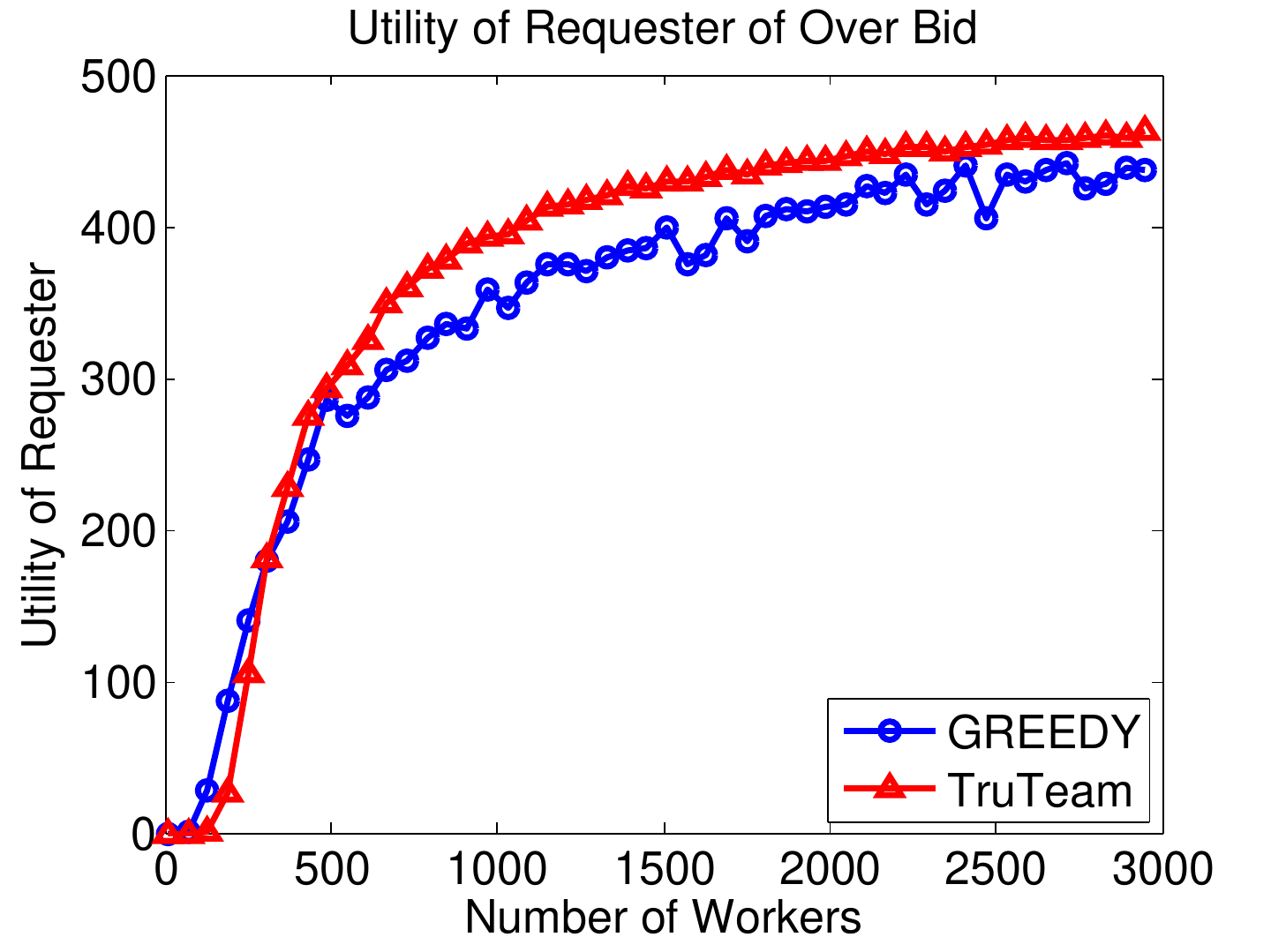}}
\hfill
\subfigure[Truthful bidding; $n=1000$]{\label{fig:utility-big-skill-truthful}\includegraphics[trim=2mm 2mm 8mm 7mm,clip,angle=0,width=0.24\textwidth]{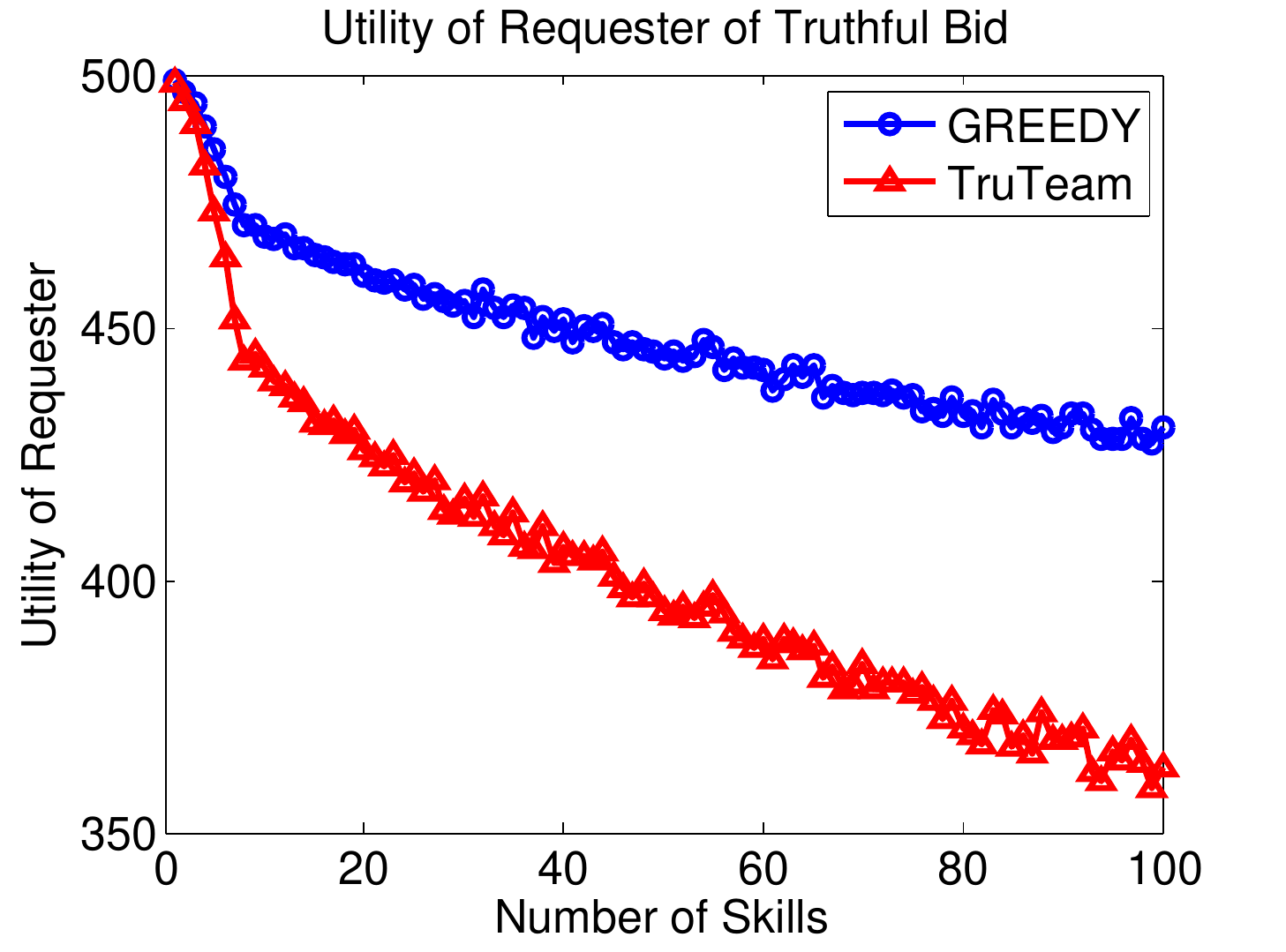}}
\hfill
\subfigure[Overbidding; $n=1000$]{\label{fig:utility-big-skill-over}\includegraphics[trim=2mm 2mm 8mm 7mm,clip,angle=0,width=0.24\textwidth]{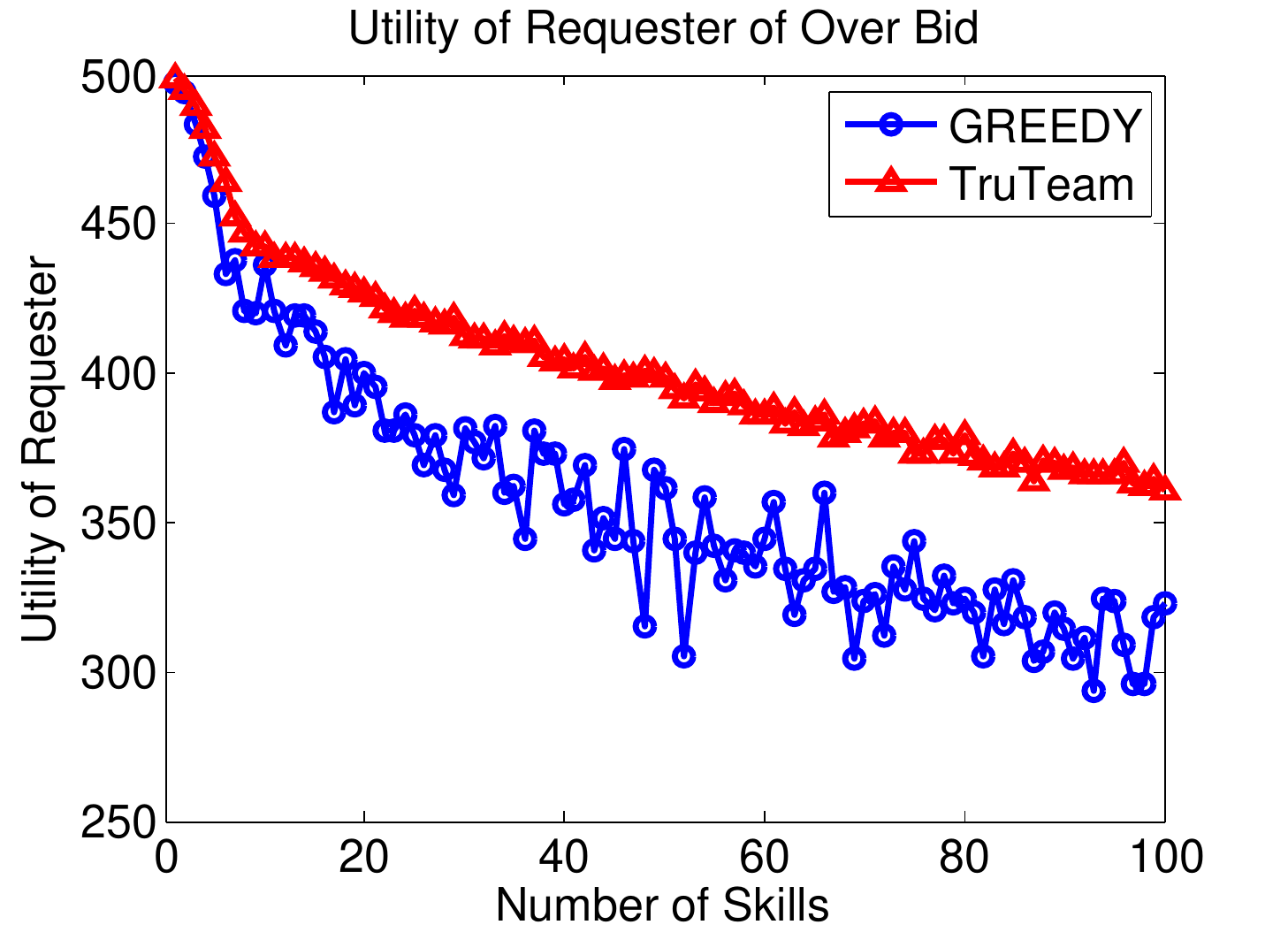}}

\caption{Simulation results on the \textsf{Large} Dataset.} \label{fig:Largedataset}
\end{figure*}

\subsubsection{Truthfulness}
Lastly, we verify the truthfulness of \textsf{TruTeam} by examining the utilities of two randomly chosen workers, $w_{637}$ and $w_{219}$. We set $n=1000$ and $l=50$, and their true costs are $c_{637}=4$ and $c_{219}=10$, respectively. In Fig.~\ref{fig:truthful-637}, we observe that $w_{637}$ is selected to perform the task if she bids her true cost $c_{637}=4$, and her utility reaches the optimal value $7$. If she overbids a value no less than 11, she is not selected and therefore her utility drops to 0. In Fig.~\ref{fig:truthful-219}, it is observed that $w_{219}$ is not selected to do the task if she bids her true cost $c_{219}=10$, and hence her utility is $0$. This is the optimal utility she can get because even though she can be selected to do the task, which only happens if she under bids (below 7), her payment will not be able to cover her true cost and hence she will receive a negative utility, as indicated in Fig.~\ref{fig:truthful-219}.

\textsf{TruTeam} ensures that it is every worker's \emph{dominant strategy} to bid her true cost in order to maximize her utility.

\begin{figure}
\subfigure[Utility of Worker $w_{637}$]{\label{fig:truthful-637}\includegraphics[trim=2mm 1.5mm 3mm 6mm,clip,angle=0,width=0.24\textwidth]{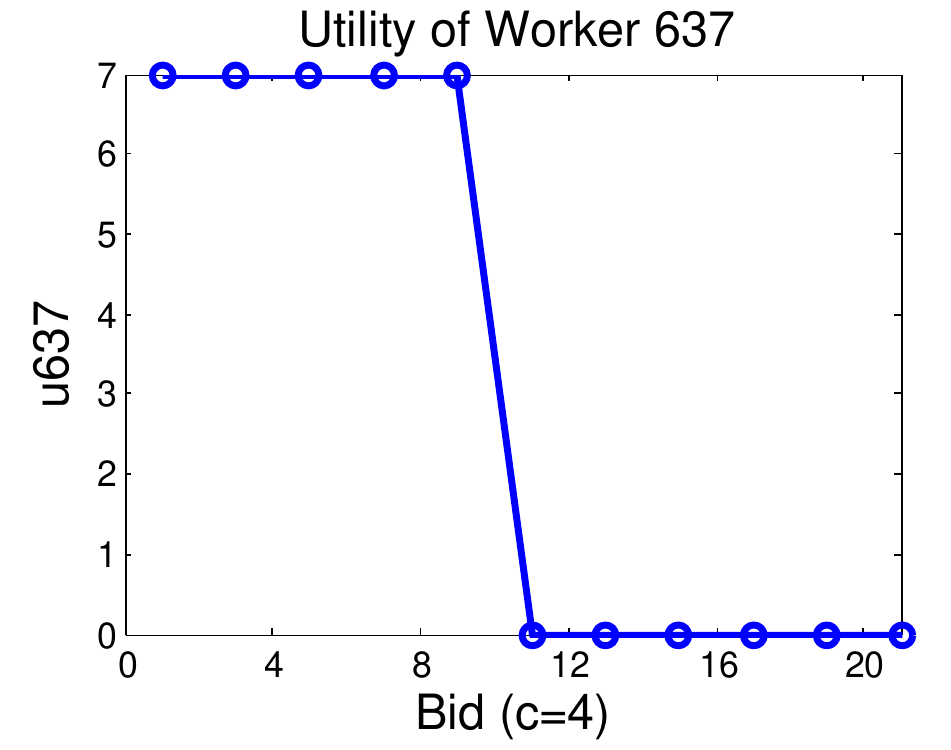}}
\hfill
\subfigure[Utility of Worker $w_{219}$]{\label{fig:truthful-219}\includegraphics[trim=2mm 1.5mm 3mm 6mm,clip,angle=0,width=0.24\textwidth]{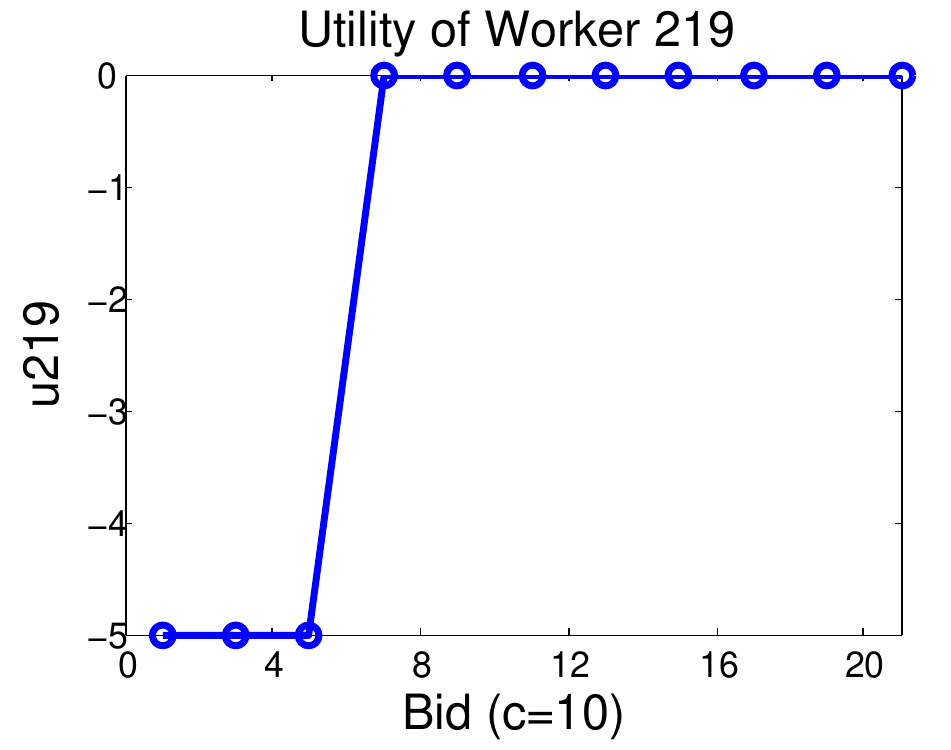}}
\caption{Workers' Utilities.} \label{fig:truthful}
\end{figure}


\section{Conclusion}
In this paper, we formulate team formation in crowdsourcing as a task allocation and pricing problem, and provide four candidate mechanisms as solutions: \textsf{OPT}, \textsf{GREEDY}, \textsf{VCG}, and \textsf{TruTeam}. We prove that although all the four mechanisms satisfy profitability and individual rationality, only \textsf{VCG} and \textsf{TruTeam} satisfy truthfulness, and only \textsf{GREEDY} and \textsf{TruTeam} are efficient. Simulation also demonstrate that \textsf{TruTeam} is the only one, among the four candidates, that is efficient, profitable, individually rational and truthful. 

To the best of our knowledge, this is the first study on team formation in collaborative crowdsourcing markets. Future research could be in the direction of taking into account the quality of contribution \cite{qcs13dcoss} and trustworthiness of workers \cite{sew14secon} when selecting and rewarding workers, or considering previous collaborations among workers and inter-dependency among multiple tasks when forming multiple teams.

\small{
\section*{Acknowledgment}
This work is supported by the French Ministry of European and Foreign Affairs under the STIC-Asia program, CCIPX project.}

%
%

%

%
\bibliographystyle{IEEEtran}
\bibliography{IEEEabrv,TruTeam}

\end{document}